\documentclass[journal]{IEEEtran}

\IEEEoverridecommandlockouts
\usepackage{amsthm}
\usepackage{mathrsfs}
\usepackage{comment}
\usepackage{siunitx}
\usepackage{placeins}
\usepackage{xcolor}
\usepackage{float}
\usepackage{bm}

\newtheorem{lemma}{Lemma}

\newtheorem{theorem}{Theorem}
\usepackage{graphicx}
\usepackage{stfloats}
\usepackage{tikz}
\usepackage{amsmath, amssymb}

\usepackage{hyperref}
\hypersetup{hidelinks}
\usepackage[capitalize]{cleveref}
\Crefname{appendix}{Appendix}{Appendices}

\usetikzlibrary{shapes.geometric, arrows.meta, positioning, calc}

\IEEEoverridecommandlockouts

\title{Two-Stage Quantum‑Classical Distribution Network Reconfiguration via Cycle-Edge Encoding\thanks{This manuscript has been authored by UT-Battelle, LLC, under Contract No. DE-AC0500OR22725 with the U.S. Department of Energy. The United States Government retains and the publisher, by accepting the article for publication, acknowledges that the United States Government retains a non-exclusive, paid-up, irrevocable, world-wide license to publish or reproduce the published form of this manuscript, or allow others to do so, for the United States Government purposes. The Department of Energy will provide public access to these results of federally sponsored research in accordance with the DOE Public Access Plan (http://energy.gov/downloads/doe-public-access-plan).}}

\author{Nowar Alashkar, Cade Kennedy, Phillip C. Lotshaw, Shaked Regev, Miguel Angel Lopez-Ruiz, Ananth Kaushik, Paul Smith, Claudio Girotto, Martin Roetteler
\thanks{N. Alashkar, C. Kennedy, and P. Smith are with Electric Power Board (EPB) of Chattanooga, Tennessee, USA 37402.}
\thanks{P. Lotshaw and S. Regev are with Computational Science and Engineering Division, Oak Ridge National Laboratory, Oak Ridge, Tennessee, USA 37830.}
\thanks{M. Lopez-Ruiz, A. Kaushik, C. Girotto, and M. Roetteler are with IonQ Inc., College Park, Maryland, USA 20740.}
}

\makeatletter
\newcommand{\linebreakand}{%
  \end{@IEEEauthorhalign}
  \hfill\mbox{}\par
  \mbox{}\hfill\begin{@IEEEauthorhalign}
}
\makeatother

\begin{document}

\maketitle
\begingroup\renewcommand\thefootnote{\textsection}

\endgroup

\begin{abstract}
    Distribution network reconfiguration (DNR) is a combinatorial optimization problem that seeks a low-loss network topology subject to topological and electrical operating constraints. In this work, we propose a two-stage quantum-classical method for DNR. First, an iterative linear-ramp Quantum Alternating Operator Ansatz (LR-QAOA) searches topologically feasible subspaces constructed using a novel cycle-edge encoding. The encoding guarantees that every represented configuration is a spanning tree while allowing the size and coverage of each subspace to be adjusted to the available qubit budget. Second, the resulting sample distribution is used to guide a mixed-integer second-order-cone programming (MISOCP) solver. The method is evaluated on six distribution networks ranging from 33 to 417 buses through ideal simulation and execution on trapped-ion quantum hardware. For the 417-bus network, quantum hardware sampling identifies a configuration within approximately $1.4\%$ of the literature reference loss. Across all six test systems, the MISOCP solver guided by hardware-derived distributions reaches an incumbent within $1\%$ of the target with a lower median solver time than the corresponding unguided solver. These findings establish a concrete framework for applying current quantum hardware to DNR, with the potential to become increasingly competitive as hardware improves.
\end{abstract}


\section{Introduction} 
\IEEEPARstart{T}{he} topology of electrical power distribution networks plays a critical role in determining system losses, as it directly influences voltage levels and current flows throughout the network. The distribution network reconfiguration (DNR) problem seeks to identify an optimal network configuration by strategically opening and closing sectionalizing and tie switches to alter power flows. A typical objective is to minimize active power losses while satisfying operational constraints, including radiality, voltage limits, and continuity of service \cite{Lotfi_2024}.  

A wide range of methods have been developed for DNR. Conventional heuristic methods, such as branch-exchange algorithms, exploit the structure of radial distribution networks to efficiently identify improving switching operations \cite{pereira2023distribution}. Metaheuristic methods, including genetic algorithms \cite{Tomoiaga_2013}, tabu search \cite{Fang_2016}, and simulated annealing \cite{5752495}, provide broader exploration of the configuration space but may require many power-flow (PF) evaluations and do not generally guarantee global optimality. Mathematical-programming approaches explicitly formulate the switching and operational constraints using mixed-integer linear \cite{11343754,10436610}, non-linear \cite{en13174440}, or second-order-cone models \cite{mahdavi2021reconfiguration,2023sunsocp,10.3389/fenrg.2023.1259445} and can provide optimality bounds for the formulated problem.

Recent work has explored hybrid quantum-classical approaches in which quantum-generated samples are used to initialize or guide classical optimization algorithms. For example, in \cite{cadavid2025scalingadvantagequantumenhancedmemetic}, configurations generated using a digitized counterdiabatic quantum optimization protocol are used to initialize a memetic tabu search for the low-autocorrelation binary-sequence problem. Similarly, in \cite{2025quantumenhancedoptimizationwarmstarts}, the authors solve maximum independent set and MaxCut problems using a state-of-the-art tabu search with initial states sampled from a Quantum Approximate Optimization Algorithm to obtain solution speedup. These approaches use quantum sampling to identify promising regions of the search space while relying on classical algorithms for subsequent refinement. Applying this strategy to DNR remains challenging because a direct switch-based encoding does not automatically enforce radiality and connectivity and may require more qubits than are available on current quantum hardware.

In this work, we propose a two-stage quantum-classical method for minimizing active power losses in DNR. The first stage uses an iterative linear-ramp Quantum Alternating Operator Ansatz (LR-QAOA) \cite{montanez2025toward} to search a sequence of topologically feasible subspaces constructed using a novel cycle-edge encoding. The encoding restricts every represented configuration to a spanning tree, thereby satisfying radiality and connectivity, while allowing each subspace to be sized according to the available qubit budget. This creates a direct tradeoff between subspace coverage and quantum hardware resources. We establish that every assignment permitted by a valid cycle-edge partition produces a unique spanning tree and that a finite collection of valid partitions can represent the complete spanning-tree space. Across iterations, the reference configuration used to construct the search space is updated as improving configurations are identified, allowing the method to explore different regions of the topologically feasible space. In the second stage, information extracted from the final sample distribution is used to guide a mixed-integer second-order-cone programming (MISOCP) solver. Specifically, the configuration with the lowest PF-evaluated line loss found is supplied as a mixed-integer programming (MIP) start, providing an initial candidate solution for the solver. Additionally, aggregate information from the distribution is used to construct variable hints, which indicate preferred values for the binary line-status variables.

The proposed method is evaluated on six distribution-network systems: 33-bus \cite{baran1989reconfiguration}, 69-bus \cite{baran1989capacitor}, 85-bus \cite{su2003reconfiguration}, 118-bus \cite{zhang2007tabu}, 136-bus \cite{guimaraes2005tabu}, and 417-bus \cite{harsh2023heuristic} using LR-QAOA samples obtained from both ideal state-vector simulation and trapped-ion quantum hardware. The iterative sampling procedure identifies improved reference configurations while maintaining a fixed 29-qubit budget across network sizes. MIP starts and variable hints derived from the final hardware distributions reduce the median solver time to reach an incumbent within $1\%$ of the literature reference loss across all six systems, with the largest reduction observed for the 417-bus network containing 473 switchable lines. Although quantum-stage overhead prevents an end-to-end runtime advantage on the tested systems, the total runtime increases more gradually across the benchmark set than the unguided solver runtime. The narrowing relative runtime gap suggests a potential crossover on larger instances, motivating further investigation beyond the present benchmark range.

The remainder of this paper is organized as follows. Section \ref{sec:methods} presents the DNR formulation, cycle-edge encoding, LR-QAOA procedure, classical evaluation and subspace-construction methods, and the iterative quantum-classical workflow. Section \ref{sec:results} describes the experimental setup and evaluates the iterative LR-QAOA procedure and distribution-guided MISOCP solver. Finally, Section \ref{sec:conclusions} summarizes the principal findings and discusses directions for future research.

\section{Methods}\label{sec:methods}
We adapt the DNR problem formulation from the mixed-integer second-order-cone programming (MISOCP) model of Mahdavi et al. \cite{mahdavi2021reconfiguration}. The model minimizes active power losses, 
\begin{equation}
    \min P_{\text{loss}} = \sum_{ij\in \Omega^\ell} R_{ij}|I_{ij}|^2,
\end{equation}
by optimizing the status of the switchable lines and associated electrical variables, where $\Omega^\ell$ is the set of distribution lines, $R_{ij}$ is the resistance of line $ij$, and $|I_{ij}|^2$ is the squared magnitude of the corresponding current. The model is subject to nodal active and reactive power balance constraints, a second-order-cone relaxation of the non-linear current-power relation, voltage and line-current limits, switch-dependent power-flow constraints, and radiality and connectivity constraints.

\subsection{Cycle-Edge Encoding} \label{sec:cycle-edge}
For a distribution network with $n$ edges, the full switch-configuration space contains $2^n$ possible configurations. However, only a small subset of these configurations satisfy the topological constraints (i.e., radiality and connectivity) \cite{macedo2018contribution}. We denote this feasible subset by $\mathcal{S}_T$.

Without constraints restricting the search to $\mathcal{S}_T$, direct sampling of the full binary switch space may devote many samples to configurations that violate the topological constraints. Formulations that enforce these constraints using auxiliary binary variables require additional qubits, further straining the resources available on current quantum hardware \cite{topo-qubo-formulation}. Therefore, we propose a cycle-edge encoding that exploits the cycle structure of the network to construct a smaller search space containing only topologically feasible configurations. Importantly, it allows us to select a structured subspace of feasible configurations, providing a controllable tradeoff between solution-space coverage and quantum hardware resources.

Let $G=(V,E)$ be a connected graph, where $V$ and $E$ are its respective sets of vertices and edges. Let $k$ denote its cyclomatic number and 
\begin{equation}
\mathcal{C}=\{C_1,C_2,\ldots,C_k\}
\end{equation}
denote a simple-cycle basis of $G$. Since two or more basis cycles may share edges, the cycle-edge sets are generally not disjoint. To eliminate overlapping edge assignments and treat the cycle sets independently, every edge belonging to multiple basis cycles is assigned to exactly one of those cycles. Of the $k$ basis cycles, $k'\leq k$ are selected for partitioning. For the remaining $k-k'$ cycles, one distinct edge per cycle is fixed open, with the fixed edges chosen outside the selected cycles so that their removal preserves connectivity. The partition of non-empty, disjoint edge sets is defined as
\begin{equation}
\mathcal{C}' = \{C_1',C_2',\ldots,C_{k'}'\},
\end{equation}
subject to 
\begin{align}
C_i' &\subseteq C_i, \qquad \text{for } i=1,\ldots,k', \label{constr:subset}\\
C_i' &\cap C_j' = \emptyset, \qquad \forall i\neq j, \label{constr:intersection}\\ 
C_i' &= C_i \setminus \bigcup_{{j<i}} C_j, \qquad \text{for }i=1,\ldots,k'. \label{constr:single-assignment}
\end{align} 
Different indexings of the basis cycles generally yield different partitions and therefore different subsets of spanning trees. This is due to the fact that different partitions allow for different edges to be open simultaneously. For example,  Fig.~\ref{fig:33bus-cycles} illustrates the overlap among the cycles in a chosen cycle-basis for the 33-bus network. The path comprising edges $(3,4)$, $(4,5)$, and $(5,6)$ is shared by cycles $C_1$, $C_2$, and $C_3$. When constructing a valid partition, this shared path is assigned to one of the corresponding partition blocks. Assigning this shared path to $C'_1$ makes its edges alternative removal choices within that block. Since exactly one edge is opened per block, opening any edge along this path excludes simultaneously opening another edge in $C'_1$, but permits it to be opened together with one edge from each of the other blocks. Choosing a different indexing can assign the shared path to a different block, changing the permitted combinations and, generally, the subset of spanning trees represented by the encoding.

\begin{figure}[htbp] 
\centering

\definecolor{cycle1}{HTML}{42BAFF}
\definecolor{cycle2}{HTML}{B494FF}
\definecolor{cycle3}{HTML}{FF0000}
\definecolor{cycle4}{HTML}{BAED91}
\definecolor{cycle5}{HTML}{E69F00}

\resizebox{\columnwidth}{!}{%
\begin{tikzpicture}[
    x=0.48cm, y=0.48cm,
    track/.style={
        line width=0.65pt,
        line cap=round,
        line join=round
    },
    tie/.style={
        track,
        dash pattern=on 2.4pt off 1.6pt
    },
    bus/.style={
        circle, fill=black,
        inner sep=0pt, minimum size=2.5pt
    },
    buslabel/.style={
        font=\fontsize{6}{7}\selectfont,
        inner sep=0.5pt, above=3pt
    },
    cyclelabel/.style={
        font=\fontsize{9}{10}\selectfont,
        inner sep=1pt
    }
]

\foreach \busnum in {1,...,18} \coordinate (b\busnum) at ({\busnum-1}, 0);
\foreach \busnum in {19,...,22} \coordinate (b\busnum) at ({\busnum-18}, 3);
\foreach \busnum in {23,...,25} \coordinate (b\busnum) at ({\busnum-21}, -4.2);
\foreach \busnum in {26,...,33} \coordinate (b\busnum) at ({\busnum-21}, -2.6);

\def \sharedbranch#1#2#3{
    \begingroup
    \edef \startbus{#1}
    \edef \endbus{#2}
    \foreach \cycleID [count=\trackIndex] in {#3}{\xdef \trackCount{\trackIndex}}

    \coordinate (waveMid) at ($(\startbus)!0.5!(\endbus)$);

    \foreach \cycleID [count=\trackIndex] in {#3}{
        \pgfmathsetmacro{\waveAmplitude}{2.4*(\trackIndex-(\trackCount+1)/2)}
        \pgfmathsetmacro{\negativeAmplitude}{-\waveAmplitude}

        \coordinate (waveA) at ($(\startbus)!0.1667!(\endbus)$);
        \coordinate (waveB) at ($(\startbus)!0.3333!(\endbus)$);
        \coordinate (waveC) at ($(\startbus)!0.6667!(\endbus)$);
        \coordinate (waveD) at ($(\startbus)!0.8333!(\endbus)$);

        \coordinate (controlA) at ($(waveA)!\waveAmplitude pt!90:(\endbus)$);
        \coordinate (controlB) at ($(waveB)!\waveAmplitude pt!90:(\endbus)$);
        \coordinate (controlC) at ($(waveC)!\negativeAmplitude pt!90:(\endbus)$);
        \coordinate (controlD) at ($(waveD)!\negativeAmplitude pt!90:(\endbus)$);

        \path let \p1 = ($(\endbus)-(\startbus)$), \n1 = {veclen(\x1,\y1)} in \pgfextra{\xdef\branchLength{\n1}};

        \ifdim \branchLength<20pt
            \coordinate (controlShort) at ($(waveC)!\waveAmplitude pt!90:(\endbus)$);
        
            \draw[track,draw=cycle\cycleID](\startbus) .. controls (controlB) and (controlShort) .. (\endbus);
        \else
            \draw[track,draw=cycle\cycleID] (\startbus) .. controls (controlA) and (controlB) .. (waveMid) .. controls (controlC) and (controlD) .. (\endbus);
        \fi}
            \endgroup
        }

\draw[black,line width=0.5pt] (b1)--(b2);

\sharedbranch{b2}{b3}{1,3}

\foreach \edgeStart in {3,4,5}{
    \pgfmathtruncatemacro{\edgeEnd}{\edgeStart+1}
    \sharedbranch{b\edgeStart}{b\edgeEnd}{1,3,5}
}

\foreach \edgeStart in {6,7}{
    \pgfmathtruncatemacro{\edgeEnd}{\edgeStart+1}
    \sharedbranch{b\edgeStart}{b\edgeEnd}{1,3,4}
}

\sharedbranch{b8}{b9}{3,4}

\foreach \edgeStart in {9,10,11}{
    \pgfmathtruncatemacro{\edgeEnd}{\edgeStart+1}
    \sharedbranch{b\edgeStart}{b\edgeEnd}{2,3,4}
}

\foreach \edgeStart in {12,13,14}{
    \pgfmathtruncatemacro{\edgeEnd}{\edgeStart+1}
    \sharedbranch{b\edgeStart}{b\edgeEnd}{2,4}
}

\foreach \edgeStart in {15,16,17}{
    \pgfmathtruncatemacro{\edgeEnd}{\edgeStart+1}
    \sharedbranch{b\edgeStart}{b\edgeEnd}{4}
}

\sharedbranch{b2}{b19}{1,3}
\sharedbranch{b19}{b20}{1,3}
\sharedbranch{b20}{b21}{1,3}
\sharedbranch{b21}{b22}{3}

\sharedbranch{b3}{b23}{5}
\sharedbranch{b23}{b24}{5}
\sharedbranch{b24}{b25}{5}

\sharedbranch{b6}{b26}{4,5}

\foreach \edgeStart in {26,27,28}{
    \pgfmathtruncatemacro{\edgeEnd}{\edgeStart+1}
    \sharedbranch{b\edgeStart}{b\edgeEnd}{4,5}
}

\foreach \edgeStart in {29,30,31,32}{
    \pgfmathtruncatemacro{\edgeEnd}{\edgeStart+1}
    \sharedbranch{b\edgeStart}{b\edgeEnd}{4}
}

\draw[tie,cycle1] (b21)--(3,1.4)--(7,1.4)--(b8);
\draw[tie,cycle2] (b9)--(8,-1.45)--(14,-1.45)--(b15);
\draw[tie,cycle3] (b22)--(11,3)--(b12);
\draw[tie,cycle4] (b18)--(17,-2.6)--(b33);
\draw[tie,cycle5] (b25)--(8,-4.2)--(b29);

\foreach \busnum in {1,...,33} \node[bus] at (b\busnum) {};

\foreach \busnum in {1,...,22,24,25,27,28,...,33} \node[buslabel] at (b\busnum) {\busnum};

\node[font=\fontsize{6}{7}\selectfont, anchor=east,inner sep=3pt] at (b23) {23};

\node[font=\fontsize{6}{7}\selectfont, anchor=east,inner sep=4pt] at (b26) {26};

\node[cyclelabel,text=cycle1] at (2,1.5)   {$C_1$};
\node[cyclelabel,text=cycle2] at (11,-0.8)   {$C_5$};
\node[cyclelabel,text=cycle3] at (8.9,1.7)   {$C_2$};
\node[cyclelabel,text=cycle4] at (15.4,-1.3) {$C_4$};
\node[cyclelabel,text=cycle5] at (3.5,-2.1)  {$C_3$};

\end{tikzpicture}%
}

\caption{Illustration of a fundamental-cycle basis for the 33-bus network. Solid connections are edges of the reference spanning tree, while dashed connections are open tie lines whose addition forms the cycles $C_1,\ldots,C_5$. Multiple colored curves along the same edge indicate that it belongs to more than one cycle; they do not represent additional physical lines.}
\label{fig:33bus-cycles}
\end{figure}

The resulting partition defines the proposed encoding. Let $x_{ij}\in \{0,1\}$ denote the binary variable associated with the $j$th edge of partition block $C'_i$, where $x_{ij}=1$ indicates that the corresponding edge is selected to be open. The encoding is represented as 
\begin{equation}
\mathbf{x} =
\underbrace{x_{11}\;x_{12}\;\cdots\;x_{1|C_1'|}}_{C_1'}
\;
\cdots
\;
\underbrace{x_{k'1}\;x_{k'2}\;\cdots\;x_{k'|C_{k'}'|}}_{C_{k'}'},
\end{equation}
subject to 
\begin{equation} \label{const:partition}
    \sum_{j=1}^{|C'_i|} x_{ij} = 1, \quad \text{for }i=1,\ldots, k'.
\end{equation}
Thus, exactly one edge is selected to be open from each of the $k'$ cycle-edge partition blocks.

\begin{theorem}\label{thm:spanning-tree}
    For any valid cycle-edge partition $\mathcal{C}'$, every assignment satisfying the cycle-edge encoding constraint in (\ref{const:partition}) produces a valid spanning-tree configuration. 
\end{theorem}

\begin{proof}[Proof of Theorem \ref{thm:spanning-tree}]
By construction, the graph obtained by excluding the $k-k'$ fixed-open edges remains connected, contains the selected cycles, and has $|V|-1+k'$ edges. A feasible assignment selects one edge $r_i\in C_i'$ from each disjoint cycle partition for removal. Since these selected edges are fixed, removing them in any order produces the same resulting graph. To establish connectivity, consider their removal starting with $C_{k'}'$ and proceeding backward to $C_1'$. By constraint (\ref{constr:single-assignment}), $r_j$ belongs to no cycle $C_i$ for $i<j$. Thus, the preceding removals leave every edge of $C_i$ in the current graph. Consequently, each selected edge remains a valid cycle-breaking removal, and connectivity is preserved throughout the sequence. After all $k'$ removals, the graph remains connected and contains $(|V|-1+k')-k'=|V|-1$ edges. Therefore, it is a spanning tree.
\end{proof}
The same argument applies component-wise to a graph with $P$ connected components. In this case, $k=|E|-|V|+P$ cycle-edge removals preserve the components and leave $|V|-P$ edges, yielding a spanning forest.

Let $\mathcal{W}(\mathcal{C}')\subseteq \mathcal{S}_T$ denote the set of spanning-tree configurations represented by a valid partition $\mathcal{C}'$. Because the partitions are disjoint, each encoded edge belongs to exactly one $C'_i$. Therefore, every combination containing one selected edge from each partition defines a unique spanning tree, and the size of the encoded subspace is 
\begin{equation}
    |\mathcal{W}(\mathcal{C}')|=\prod_{i=1}^{k'} |C'_i|,
\end{equation}
subject to (\ref{const:partition}). The number of binary variables $x_{ij}$ required to encode all $k'$ blocks of $\mathcal{C}'$ is $\sum_{i=1}^{k'} |C'_i|$.

A single partition generally represents only a subset of $\mathcal{S}_T$. However, the encoding is collectively capable of representing the complete spanning-tree space through different choices of cycle bases and valid partitions. 

\begin{theorem}
\label{thm:reachability}
Let $\mathcal{T}(G)=\{T_1,T_2,\ldots,T_t\}$ denote the set of all spanning trees of a connected
graph $G$. There exists a finite collection of valid partitions
\begin{align}
    &\mathcal{P} = \{\mathcal{C}'_1,\mathcal{C}'_2,\ldots,\mathcal{C}'_m\}\\
    \text{s.t.}\; &
    \mathcal{W}(\mathcal{P}) := \bigcup_{i=1}^{m} \mathcal{W}(\mathcal{C}'_i) = \mathcal{T}(G).
\end{align}
\end{theorem}

The proof of theorem \ref{thm:reachability} is provided in Appendix \ref{app:spanning-tree-reachability}. Together, theorems \ref{thm:spanning-tree} and \ref{thm:reachability} establish that the encoding is topologically valid and collectively capable of representing the complete topologically feasible space. In practice, the proposed method does not attempt to identify all $\mathcal{C}'_i\in \mathcal{P}$ required to cover $\mathcal{S}_T$. Instead, it iteratively constructs and searches selected subspaces, each restricted according to the available qubit budget, as discussed in Section \ref{sec:subspace}. This allows the method to explore different regions of $\mathcal{S}_T$ across iterations while ensuring that every encoded configuration satisfies the topological constraints.

\begin{figure*}[htbp] 
\centering
\begin{tikzpicture}[
    wblock/.style={
        draw,
        fill=white,
        minimum width=1.15cm,
        minimum height=1.1cm,
        font=\small
    },
    cblock/.style={
        draw,
        fill=white,
        text width=1.55cm,
        minimum height=4.35cm,
        align=center,
        font=\small
    },
    mblock/.style={
        draw,
        fill=white,
        text width=1.65cm,
        minimum height=1.1cm,
        align=center,
        font=\small
    },
    measure/.style={
        draw,
        fill=white,
        minimum width=1.05cm,
        minimum height=1cm,
        font=\small
    },
    dots/.style={
        fill=white,
        inner xsep=5pt,
        inner ysep=1pt
    },
    statecut/.style={
        red,
        dashed,
        thick
    },
    statelabel/.style={
        black,
        font=\small
    }
]

    \def\xstart{0}
    \def\xw{1}
    \def\xinitialstate{2.1}
    \def\xc{3.5}
    \def\xm{5.55}
    \def\xdots{7.5}
    \def\xcfinal{9.45}
    \def\xmfinal{11.5}
    \def\xfinalstate{13.0}
    \def\xmeasure{14.1}
    \def\xend{15}
    
    \def\ytop{0}
    \def\ymiddle{-1.5}
    \def\ybottom{-3.25}
    
    \pgfmathsetmacro{\ycenter}{(\ytop+\ybottom)/2}
    \pgfmathsetmacro{\ydots}{(\ymiddle+\ybottom)/2+0.1}
    
    \draw[thick] (\xstart+0.2,\ytop) -- (\xend,\ytop);
    \draw[thick] (\xstart+0.2,\ymiddle) -- (\xend,\ymiddle);
    \draw[thick] (\xstart+0.2,\ybottom) -- (\xend,\ybottom);
    
    \node[left] at (\xstart,\ytop) {$C'_1$};
    \node[left] at (\xstart,\ymiddle) {$C'_2$};
    \node at (-0.35,\ydots) {$\vdots$};
    \node[left] at (\xstart,\ybottom) {$C'_{k'}$};
    
    \node at (\xstart+0.1,\ytop) {$\Bigg\{$};
    \node at (\xstart+0.1,\ymiddle) {$\Bigg\{$};
    \node at (\xstart+0.1,\ybottom) {$\Bigg\{$};
    
    \node[wblock] at (\xw,\ytop) {$W_1$};
    
    \node[wblock] at (\xw,\ymiddle) {$W_2$};
    
    \node at (\xw,\ydots) {$\vdots$};
    
    \node[wblock] at (\xw,\ybottom) {$W_{k'}$};
    
    \draw[statecut] (\xinitialstate,\ybottom-0.6) -- (\xinitialstate,\ytop+0.6);
    \node[statelabel, below] at (\xinitialstate,\ybottom-0.6) {$|\Psi_0\rangle$};
    
    \node[cblock] at (\xc,\ycenter) {$U_C(\gamma_0)$};
    
    \node[mblock] at (\xm,\ytop) {$U_{M_1}(\beta_0)$};
    \node[mblock] at (\xm,\ymiddle) {$U_{M_2}(\beta_0)$};
    \node at (\xm,\ydots) {$\vdots$};
    \node[mblock] at (\xm,\ybottom) {$U_{M_{k'}}(\beta_0)$};
    
    \node[dots] at (\xdots,\ytop) {$\cdots$};
    \node[dots] at (\xdots,\ymiddle) {$\cdots$};
    \node[dots] at (\xdots,\ybottom) {$\cdots$};
    \node[font=\scriptsize, above=5pt] at (\xdots,\ytop) {$p$ layers};
    
    \node[cblock] at (\xcfinal,\ycenter) {$U_C(\gamma_{p-1})$};
    
    \node[mblock] at (\xmfinal,\ytop) {$U_{M_1}(\beta_{p-1})$};
    \node[mblock] at (\xmfinal,\ymiddle) {$U_{M_2}(\beta_{p-1})$};
    \node at (\xmfinal,\ydots) {$\vdots$};
    \node[mblock] at (\xmfinal,\ybottom) {$U_{M_{k'}}(\beta_{p-1})$};
    
    \draw[statecut] (\xfinalstate,\ybottom-0.6) -- (\xfinalstate,\ytop+0.6);
    
    \node[statelabel, below] at (\xfinalstate,\ybottom-0.6) {$|\Psi\rangle$};
    
    \node[measure] at (\xmeasure,\ytop) {$\mathcal{M}$};
    \node[measure] at (\xmeasure,\ymiddle) {$\mathcal{M}$};
    \node at (\xmeasure,\ydots) {$\vdots$};
    \node[measure] at (\xmeasure,\ybottom) {$\mathcal{M}$};

\end{tikzpicture}

\caption{Block-level LR-QAOA circuit. The initial state $|\Psi_0\rangle$ is prepared by applying the $W$-state preparation operator $W_i$ to the register associated with each partition block $C'_i$. Each of the $p$ LR-QAOA layers consists of a cost unitary $U_C(\gamma_j)$ followed by partition-block mixers $U_{M_i}(\beta_j)$. The mixers act independently on each register to preserve the partition constraint in (\ref{const:partition}). The final state $|\Psi\rangle$ is measured in the computational basis.}
\label{fig:block_lr_qaoa}
\end{figure*}

\subsection{LR-QAOA}
For a valid cycle-edge partition $\mathcal{C}'$, let $n_i=|C'_i|$ denote the number of encoded edges in partition block $C'_i$. Because the cycle-edge encoding requires exactly one edge to be open in each partition block, the corresponding qubits are initialized in the $n_i$-qubit $W$ state
\begin{equation}
    |W_i\rangle = \frac{1}{\sqrt{n_i}} \sum_{a=1}^{n_i} |0\cdots1_a \cdots0\rangle.
\end{equation}
The complete initial state is therefore the tensor product of the partition-block $W$ states,
\begin{equation}
    |\Psi_0\rangle = \bigotimes_{i=1}^{k'} |W_i\rangle.
\end{equation}
Under ideal state preparation, this construction ensures that $|\Psi_0\rangle$ is a superposition only of states corresponding to configurations in $\mathcal{W}(\mathcal{C}')\subseteq \mathcal{S}_T$.

The block-level LR-QAOA circuit is shown in Fig. \ref{fig:block_lr_qaoa}. Starting from $|\Psi_0\rangle$, LR-QAOA alternates between a cost unitary $U_C(\gamma_j)$ that encodes the approximate line-loss objective and mixer unitaries $U_{M_i}(\beta_j)$ that redistribute probability within $\mathcal{W}(\mathcal{C}')$, where $\gamma_j$ and $\beta_j$ are circuit parameters specified by the linear-ramp schedule. For a circuit containing $p$ layers, the final state is constructed by
\begin{equation}
    |\Psi\rangle = \prod_{j=0}^{p-1} \left(\bigotimes_{i=1}^{k'} U_{M_i}(\beta_j)\right)U_C(\gamma_j)|\Psi_0\rangle.
\end{equation}

The cost unitary is derived from a quadratic surrogate of the line-loss objective, 
\begin{equation} \label{eq:surrg}
f(\mathbf{x})= \sum_{1\leq i\leq k'}  q_{a}^{(i)} x_{ia} + \sum_{1\leq i< \ell \le k'}q^{(i\ell)}_{ab}x_{ia}x_{\ell b}, \end{equation}
where repeated edge indices $a$ and $b$ are implicitly summed over their respective partition blocks. Only quadratic interactions between different partition blocks are included here because the one-hot constraint makes within-block quadratic terms reducible to constant and linear terms. The coefficients $q_a^{(i)}$ and $q_{ab}^{(i\ell)}$ are fitted by lasso regression \cite{tibshirani1996regression} to line losses obtained by evaluating configurations sampled from $\mathcal{W}(\mathcal{C}')$. The resulting quadratic model is then mapped to a diagonal cost Hamiltonian $H_C$ by replacing each binary variable $x_{ia}$ with $(\mathbb{I}-Z_{ia})/2$, where $Z_{ia}$ denotes the Pauli-$Z$ operator acting on the $a$th qubit of block $C'_i$, such that 
\begin{equation}
    H_C|\mathbf{x}\rangle = f(\mathbf{x})|\mathbf{x}\rangle.
\end{equation}
Before constructing the cost unitary, $H_C$ is rescaled to obtain a normalized Hamiltonian  $\widetilde{H}_C$, following the procedure described in \cite{montanez2025toward}. The cost unitary at layer $j$ is then given by
\begin{equation}
    U_C(\gamma_j) = e^{-i\gamma_j\widetilde{H}_C}.
\end{equation}

To ensure that the evolution remains within $\mathcal{W}(\mathcal{C}')$, a constraint-preserving $XY$-ring mixer is used \cite{wang2020xy} for each partition block. For block $C'_i$, the mixer Hamiltonian is 
\begin{equation}
    H_{M_i} = \sum_{a=1}^{n_i} (X_{ia}X_{ia+1} + Y_{ia} Y_{ia+1}),
\end{equation}
where $X_{ia}$ and $Y_{ia}$ are Pauli-$X$ and Pauli-$Y$ operators that act on the $a$th qubit of $C'_i$, and $a+1$ is evaluated modulo $n_i$. The corresponding partition-block mixer unitary is 
\begin{equation}
    U_{M_i}(\beta_j) = e^{i\beta_jH_{M_i}}.
\end{equation}

Rather than optimizing $2p$ circuit parameters $\gamma_j$ and $\beta_j$, we fix them according to the linear-ramp schedule 
\begin{equation} \beta_{j} = \left( 1 - \frac{j}{p} \right) \Delta_{\beta}, \qquad \gamma_{j} = \frac{j+1}{p} \Delta_{\gamma}, \end{equation}
where $j=0,\ldots,p-1$. Preliminary parameter sweeps on small test networks, including the 33-bus system, informed the choice of $\Delta_\beta=0.2$ and $\Delta_\gamma=1$. Both values were then held fixed for all reported experiments. We chose the normalization factor for the cost Hamiltonian so that, with $\Delta_\gamma=1$, the rotation angles associated with each term are bounded in magnitude by $\pi/2$.

\subsection{Classical Evaluation}\label{sec:classical-eval}
We evaluate candidate configurations using the backward-forward sweep (BFS) PF method \cite{shirmohammadi1988compensation}, which is well suited to radial distribution networks. The backward sweep accumulates the currents from the terminal buses toward the source, while the forward sweep updates the bus voltages from the source toward the terminal buses. These sweeps are repeated until convergence, after which the total active power loss of each configuration is calculated. The computational efficiency of this method enables lightweight classical pre- and post-processing. We use these PF evaluations to generate training data used to fit the quadratic surrogate in (\ref{eq:surrg}). Additionally, PF is used to evaluate a selected subset of configurations sampled by the LR-QAOA and validate MISOCP incumbents. 

\subsection{Subspace Construction}\label{sec:subspace}
The construction of each LR-QAOA subspace begins from a reference spanning-tree configuration $T_t$, with $T_0$ corresponding to the base configuration of the network. Adding one edge to $T_t$ creates a unique fundamental cycle. Repeating this process for all $k$ open edges defines the fundamental-cycle basis $\mathcal{C}_t$, which is converted into a valid cycle-edge partition $\mathcal{C}_t'$ according to the conditions in Section \ref{sec:cycle-edge}. If representing the complete partition would exceed the available qubit budget, only a subset of the partitions and their associated edges is encoded. We denote the number of edges retained in this restricted encoding by $n_q$, with one qubit assigned to each retained edge. The current open edge is retained in each selected block, while the open edges associated with unselected blocks remain fixed. Consequently, $T_t$ remains representable within the encoded subspace and is available to the quantum search if the subspace contains no improving configurations. 

A heuristic procedure is used to select the partition blocks and their associated edges under a restricted qubit budget. For each partition block, one trial configuration is generated by replacing its current open edge with a randomly selected alternative edge from the same block. The resulting configuration is evaluated using a single BFS PF calculation, and the partition blocks are ranked in ascending order of active power loss. The highest-ranked $k'$ blocks are selected subject to the budget. The available qubits are then distributed as evenly as possible among the selected blocks, and the remaining candidate edges are chosen at approximately uniform intervals along their corresponding cycles. This spreads the encoded edges across the selected cycles and increases the topological coverage of the resulting subspace.

\subsection{Two-Stage Quantum-Classical Workflow}

Fig.~\ref{fig:hybrid_workflow} summarizes the two-stage quantum-classical workflow. At iteration $t$, the current reference configuration $T_t$ is used to construct the cycle-edge subspace explored by the LR-QAOA, as described in Section \ref{sec:subspace}. We then sample the resulting circuit to produce a set of candidate network configurations. 

\begin{figure}[htbp]
\centering
\resizebox{\columnwidth}{!}{%
\begin{tikzpicture}[
    x=1cm,
    y=1cm,
    every node/.style={font=\small},
    box/.style={
        draw,
        rounded corners=2pt,
        minimum width=2.5cm,
        minimum height=1.0cm,
        align=center,
        fill=white
    },
    flow/.style={
        draw,
        -latex,
        thick
    }
]

\node[box] (tree) at (-4,0) { \textbf{Spanning Tree}\\ 
    $T_t$\\
};
\node[draw, circle, fill=white, inner sep=1.5pt, font=\scriptsize] at ([xshift=-2pt,yshift=2pt]tree.north west) {1.};

\node[box] (encoding) at (0,0) {\textbf{Cycle-Edge Partition}\\
$\mathcal{C}'_t$\\
};
\node[draw, circle, fill=white, inner sep=1.5pt, font=\scriptsize] at ([xshift=-2pt,yshift=2pt]encoding.north west) {2.};

\node[box] (qaoa) at (4,0) {\textbf{LR-QAOA}\\
    Distribution sampling\\
};
\node[draw, circle, fill=white, inner sep=1.5pt, font=\scriptsize] at ([xshift=-2pt,yshift=2pt]qaoa.north west) {3.};

\node[box] (samples) at (0,-2.0) {\textbf{Distribution} \\ \textbf{Filter}
};
\node[draw, circle, fill=white, inner sep=1.5pt, font=\scriptsize] at ([xshift=-2pt,yshift=2pt]samples.north west) {4.};

\node[box] (misocp) at (0,-4.0) {\textbf{MISOCP}\\ 
MIP start + variable hints
};
\node[draw, circle, fill=white, inner sep=1.5pt, font=\scriptsize] at ([xshift=-2pt,yshift=2pt]misocp.north west) {5.};

\draw[flow] (tree.east) -- (encoding.west);

\draw[flow] (encoding.east) -- (qaoa.west);

\draw[flow] (qaoa.south) -- (4,-2.0) -- (samples.east);

\draw[flow] (samples.west) -- (-4.0,-2.0) -- (-4.0,-1.0) -- (tree.south);

\node[font=\footnotesize, align=center] at (-2.4,-2.35){
    $T_{t+1}$\\
};

\draw[flow] (samples.south) -- (misocp.north);

\node[font=\footnotesize, anchor=west] at (0.2,-2.99) {Final iteration};

\end{tikzpicture}%
}

\caption{High-level overview of the two-stage quantum-classical workflow. Starting from the current spanning-tree configuration $T_t$, a fundamental-cycle basis is constructed and converted into a cycle-edge partition $\mathcal{C}'_t$ that defines the LR-QAOA search space. Samples from the LR-QAOA circuit pass through the distribution filter, which removes topologically infeasible outcomes and ranks the remaining configurations using the quadratic surrogate. The top-ranked candidates are evaluated by BFS PF to obtain $T_{t+1}$. After the final iteration, the lowest-loss evaluated configuration provides a MIP start, while information from the final samples provides variable hints for the MISOCP solver.}
\label{fig:hybrid_workflow}
\end{figure}

Although the encoding preserves topological feasibility under ideal execution, hardware noise can produce bit strings that violate these constraints. In the distribution filter stage, we remove such outcomes and rank the retained configurations using the quadratic surrogate in (\ref{eq:surrg}) as an inexpensive approximation of their relative line losses. We then evaluate the top 100 ranked candidates using the BFS PF method, reducing the number of PF calculations required at each iteration. 

Among the candidates evaluated at iteration $t$, the configuration with the lowest line loss is selected as $T_{t+1}$. This configuration is then used to construct the cycle-edge subspace for the next LR-QAOA iteration, and the process is repeated for the prescribed number of iterations.

After the final LR-QAOA iteration, the lowest-loss configuration evaluated from the final distribution supplies the line-status assignment as a MIP start for the MISOCP solver. All unique open edges in the top 10 candidate configurations are additionally used to construct variable hints indicating preferred open states for the corresponding variables. These hints guide the solver's search without fixing the variables or restricting the feasible set.

\section{Results}\label{sec:results}
The proposed method is evaluated on six distribution networks ranging from 33 to 417 buses, with the corresponding MISOCP formulations containing between 37 and 473 binary line-status variables. For all networks, the LR-QAOA circuits use $n_q=29$ qubits and $p=2$ layers. We assess the effect of the hardware-derived sample distributions by comparing guided and unguided MISOCP solves using Gurobi \cite{gurobi_optimizer}. For each network, the two sets of runs use the same formulation, solver settings, random seeds, time limits, and computing environment; only the guided runs receive a MIP start and variable hints derived from the final LR-QAOA distribution. 

In Section \ref{sec:sampling-results}, the reported line-loss values are obtained using the BFS PF method discussed in Section \ref{sec:classical-eval}; evaluating the full sample distributions serves only to analyze the results. In Section \ref{sec:misocp-results}, the plotted incumbent values are the MISOCP objective values reported by Gurobi. Because these values can differ slightly from BFS PF-evaluated losses, each incumbent counted as reaching the $1\%$ target is validated using BFS PF. We define the $1\%$ target as a BFS PF-evaluated loss no greater than $1\%$ above the literature reference loss for the corresponding network.

We first examine the evolution of the LR-QAOA sample distributions under ideal state-vector simulation and execution on trapped-ion quantum hardware (see Appendix \ref{app:quantum_hardware}). We then compare the guided and unguided MISOCP solvers using the hardware-derived distributions, measuring the time required to reach an incumbent within $1\%$ of the literature reference loss. These reference values are $139.55$, $99.62$, $469.88$, $869.70$, $280.19$, and $581.57$ $\mathrm{kW}$ for the 33-, 69-, 85-, 118-, 136-, and 417-bus networks, respectively \cite{review, cortes2024innovative}. These published values serve as fixed benchmarks for solution quality; they are not asserted to be certified optima. 

\subsection{Iterative LR-QAOA Sampling} \label{sec:sampling-results}
\begin{figure*}[htbp]
    \centering
    \includegraphics[width=\textwidth]{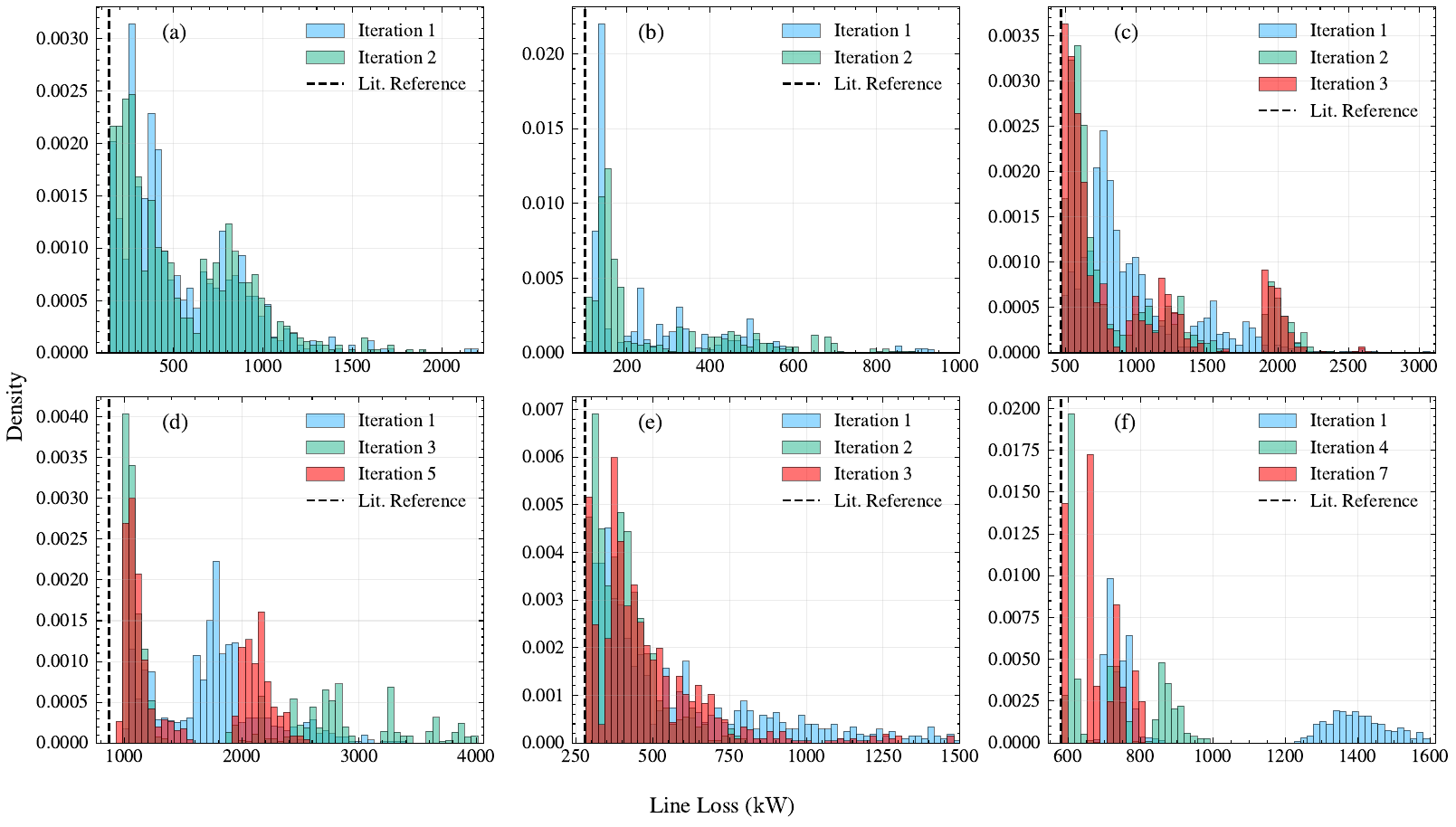}
    \caption{Line-loss distributions of samples obtained from ideal state-vector simulation of LR-QAOA for the six distribution-network test systems: (a) 33-bus, (b) 69-bus, (c) 85-bus, (d) 118-bus, (e) 136-bus, and (f) 417-bus. Each panel shows the distributions obtained at selected iterations of the LR-QAOA procedure, with the dashed vertical line indicating the literature reference loss for the corresponding system. Each LR-QAOA circuit used $p=2$ layers and 1000 shots. Across the test networks, later iterations generally yield more samples in the lower-loss range, illustrating the progression of the iterative LR-QAOA procedure.
}
    \label{fig:simulator_dist_all}
\end{figure*}

Under the 29-qubit budget in these experiments, each encoded subspace includes $k'=5$ partition blocks for the 33-bus and 69-bus networks, both of which have a cyclomatic number of $k=5$. The 85-bus, 118-bus, 136-bus, and 417-bus networks have cyclomatic numbers of $k=13, 15, 21, \text{ and } 59$, respectively. For each of these systems, $k'=8$ partition blocks are selected. Eight blocks were chosen to balance the number of blocks allowed to vary within each iteration against the number of candidate edges retained per block. Including more blocks permits simultaneous changes across more blocks, but leaves fewer qubits available to represent alternative edge choices within each block. The blocks and their retained edges are selected using the procedure described in Section \ref{sec:subspace}.

Fig.~\ref{fig:simulator_dist_all} and Table~\ref{tab:sim} summarize the LR-QAOA samples obtained under ideal state-vector simulation. The figure shows the full BFS PF evaluated line-loss distributions at selected iterations, while the table reports the lowest line loss identified at each iteration. Across the six systems, the best sampled configurations generally improve as the reference configuration is updated and the encoded subspace is reconstructed. The 33-bus, 69-bus, 85-bus, 118-bus, 136-bus, and 417-bus systems all show decreasing best sampled losses over the displayed iterations. However, since the reference update and subspace-selection procedure are heuristic, monotonic improvement is not guaranteed.

\begin{table}[htbp]
    \centering
    \caption{Best BFS PF line loss ($\mathrm{kW}$) sampled at each LR-QAOA iteration under ideal state-vector simulation.}
    \renewcommand{\arraystretch}{1.15}
    \setlength{\tabcolsep}{3.5pt}
    \begin{tabular}{lccccccc}
    \hline
         Network & Iter. 1 & Iter. 2 & Iter. 3 & Iter. 4 & Iter. 5 & Iter. 6 & Iter. 7 \\
    \hline
        33-bus & 147.81 & 142.68 & -- & -- & --& -- & -- \\ 
        69-bus & 99.82 & 99.71 & -- & -- & --& -- & -- \\
        85-bus & 480.07 & 477.10 & 475.92 & -- & -- & -- & --\\ 
        118-bus & 1010.89 & 982.35 & 981.33 & 974.37 & 971.23 & -- & --\\ 
        136-bus & 302.53 & 286.50 & 285.50 & -- & -- & -- & --\\ 
        417-bus & 681.32 & 668.48 & 605.49 & 594.18 & 591.04 & 586.69 & 586.66 \\ 
    \end{tabular}
    \label{tab:sim}
\end{table}

The improvement in the best configuration is most pronounced for the larger instances. In the 417-bus network, the best sampled line loss decreases from 681.32 $\mathrm{kW}$ in the first iteration to $586.66~\mathrm{kW}$ by the seventh iteration, reaching within $1\%$ of the literature reference loss. The 118-bus system also shows a substantial reduction over the iterations, decreasing from $1010.89$ to $971.23~\mathrm{kW}$. For the 85-bus and 136-bus systems, the final iterations reach $475.92~\mathrm{kW}$ and $285.50~\mathrm{kW}$, respectively. The smaller 33-bus and 69-bus networks show more modest absolute changes because their initial iterations already identify low line-loss configurations near the literature reference loss, leaving less room for large subsequent reductions in line loss.

The same iterative sampling procedure was then executed on the IonQ Forte Enterprise QPU to assess whether the observed behavior persists under hardware execution. Because hardware noise can produce measured bit strings outside the encoded feasible space, the hardware distributions in Fig.~\ref{fig:qpu_dist_all} and the best sampled losses in Table~\ref{tab:qpu} are computed after retaining topologically feasible configurations and evaluating them with BFS PF. Despite the noise, the best retained samples improve across iterations for the 69-bus, 85-bus, 118-bus, 136-bus, and 417-bus networks. The 417-bus hardware run decreases from $682.10~\mathrm{kW}$ in the first iteration to $589.56~\mathrm{kW}$ by the seventh iteration, reaching within approximately $1.4\%$ of the literature reference loss. The 33-bus case is the main exception, with the best retained sample increasing slightly between the two iterations. 

\begin{figure*}[htbp]
    \centering
    \includegraphics[width=\textwidth]{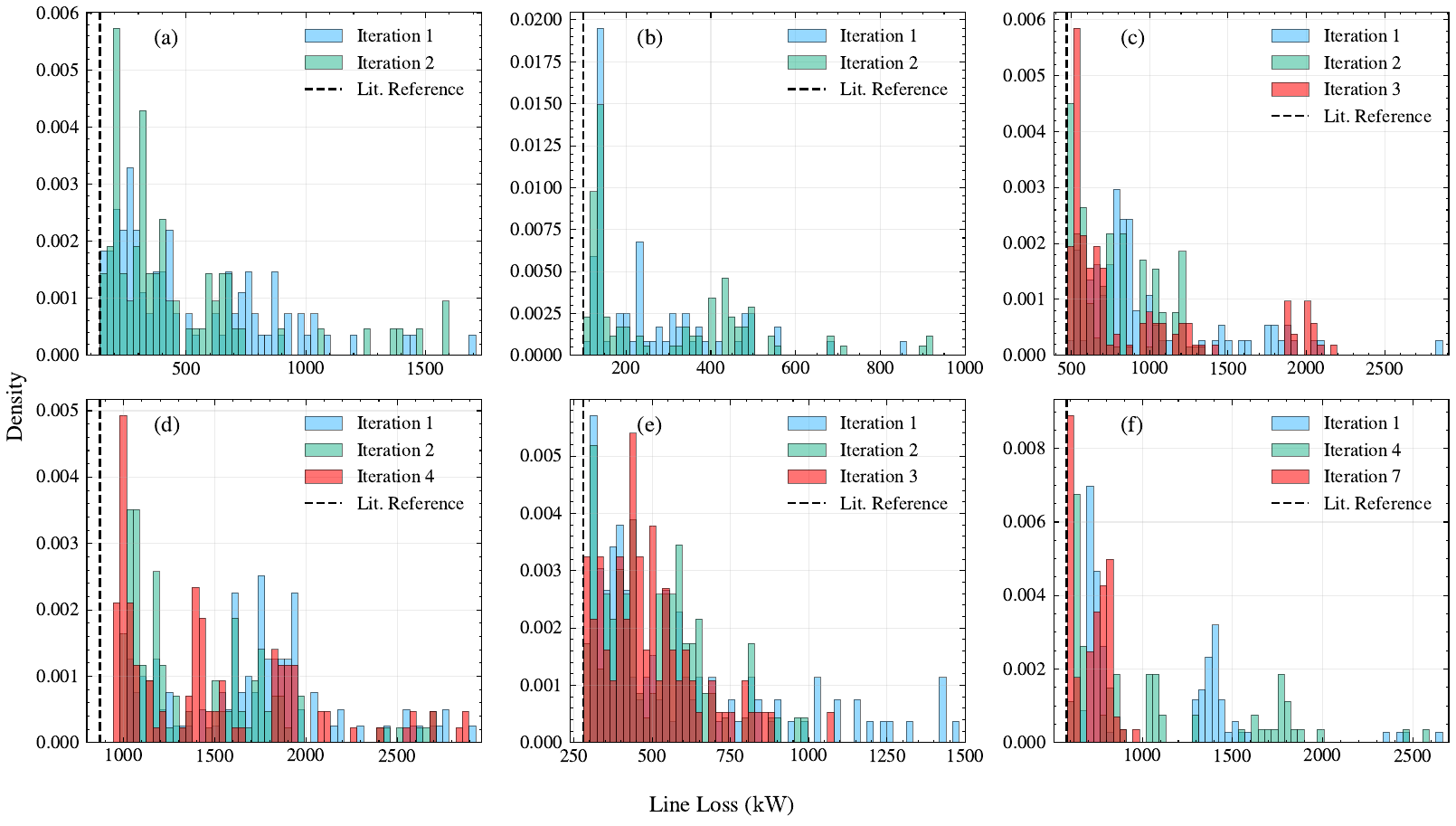}
    \caption{Line-loss distributions of samples obtained from LR-QAOA execution on the IonQ Forte Enterprise QPU for the six distribution-network test systems: (a) 33-bus, (b) 69-bus, (c) 85-bus, (d) 118-bus, (e) 136-bus, and (f) 417-bus. Each panel shows the distributions obtained at selected iterations of the LR-QAOA procedure, with the dashed vertical line indicating the literature reference loss for the corresponding system. Each LR-QAOA circuit used $p=2$ layers and 1000 shots.
}
    \label{fig:qpu_dist_all}
\end{figure*}

\begin{table}[htbp]
    \centering
    \caption{Best BFS PF line loss ($\mathrm{kW}$) from retained QPU samples at each LR-QAOA iteration.}
    \renewcommand{\arraystretch}{1.15}
    \setlength{\tabcolsep}{3.5pt}
    \begin{tabular}{lccccccc}
    \hline
         Network & Iter. 1 & Iter. 2 & Iter. 3 & Iter. 4 & Iter. 5 & Iter. 6 & Iter. 7 \\
    \hline
        33-bus & 153.49 & 155.80 & -- & -- & --& -- & -- \\ 
        69-bus & 111.67 & 103.98 & -- & -- & --& -- & -- \\
        85-bus & 509.80 & 477.73 & 475.79 & -- & -- & -- & --\\ 
        118-bus & 1018.39 & 992.18 & 978.21 & 959.00 & -- & -- & --\\ 
        136-bus & 302.72 & 295.14 & 285.35 & -- & -- & -- & --\\ 
        417-bus & 682.10 & 647.62 & 631.44 & 612.98 & 598.60 & 591.75 & 589.56 \\

    \end{tabular}
    \label{tab:qpu}
\end{table}

The QPU distributions differ from ideal simulations in relative frequency and spread, consistent with hardware noise and subspace differences. Nevertheless, the retained hardware samples continue to include low-loss configurations across all six networks, with particularly strong improvement for the 417-bus case. Thus, the hardware experiment supports the use of the final QPU sample distribution as a source of problem-specific information for the guided MISOCP solve, rather than only as a qualitative demonstration of circuit execution.

\subsection{Distribution-Guided MISOCP} \label{sec:misocp-results}

We next evaluate whether the information contained in the final QPU LR-QAOA sample distribution improves the subsequent MISOCP solve. For each network, the guided solver is initialized using the lowest-loss configuration identified from the final distribution and receives variable hints for the line-status variables based on the same distribution, while the unguided solver receives neither.

Fig.~\ref{fig:convergence-summary} compares the time required for the guided and unguided MISOCP solvers to reach an incumbent solution within $1\%$ of the literature reference loss. Across all six networks, the guided solver has a lower median time to reach the target than the unguided solver, indicating that the LR-QAOA-derived MIP start and variable hints provide useful information for the classical search. The improvement is larger for larger systems, where the unguided solver requires substantially more time to identify a comparable incumbent. For the 417-bus network, the median solver time decreases from $1415.47~\mathrm{s}$ unguided to $34.71~\mathrm{s}$ guided.

\begin{figure}[htbp]
    \centering
    \includegraphics[width=\linewidth]{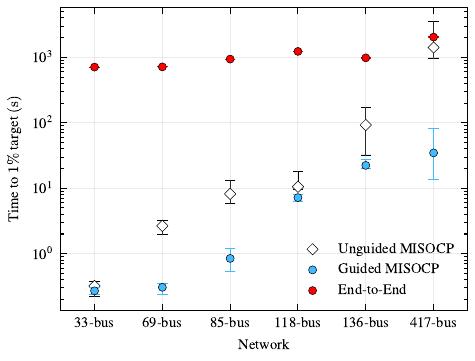}
    \caption{Time required to reach an incumbent solution within $1\%$ of the literature reference loss for the six distribution-network test systems. Unguided MISOCP denotes the standalone solver and guided MISOCP uses the LR-QAOA-derived MIP start and variable hints. End-to-end runtime includes the LR-QAOA stage and guided MISOCP optimization. Markers indicate the median across 20 solver seeds, with error bars spanning the interquartile range.
}
    \label{fig:convergence-summary}
\end{figure}

The end-to-end runtimes in Fig.~\ref{fig:convergence-summary} add QPU execution and classical pre- and post-processing times to those same guided solves. For the smaller and medium-sized networks, QPU execution dominates the total runtime, so the reduction in solver time does not produce an end-to-end speedup. For the 417-bus network, however, the end-to-end time approaches the unguided MISOCP time to target despite the added quantum-stage overhead. Thus, the guided solver requires less solver time than the unguided solves, with the total runtime gap narrowing for the largest network.

Fig.~\ref{fig:runtime-comparison} shows the same time-to-target metric as in Fig.~\ref{fig:convergence-summary} as a function of the number of switchable lines, which corresponds to the number of binary line-status variables in the MISOCP formulation. The unguided MISOCP runtime increases sharply with the number of lines, reflecting the growing difficulty of the search. In contrast, the end-to-end runtime varies more gradually across the tested systems because the LR-QAOA stage requires the majority of the execution time and remains within a similar resource range across all networks. As a result, the gap between the unguided MISOCP runtime and the end-to-end runtime narrows for the largest network, suggesting a potential crossover regime in which the cost of obtaining quantum-derived guidance becomes comparable to the time saved in the classical search. For sufficiently difficult instances, the reduction in guided MISOCP solver time relative to the unguided solver may grow enough to offset the time spent on quantum execution and classical pre- and post-processing.  This could yield an end-to-end runtime advantage over the unguided solver, and testing the precise scaling at large sizes is an important future direction. The current comparison should therefore be interpreted as an empirical trend over the benchmark set rather than an asymptotic scaling claim, but it suggests that distribution-derived solver guidance becomes more relevant as the classical search space grows. 

\begin{figure}[htbp]
    \centering
\includegraphics[width=\linewidth]{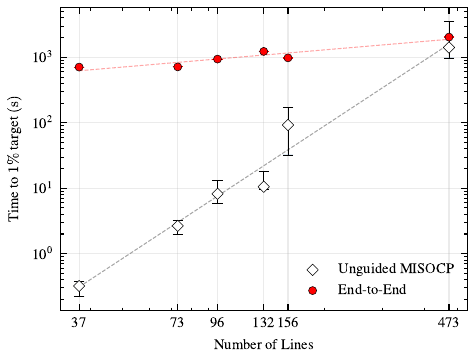}
    \caption{Time required to reach an incumbent solution within $1\%$ of the literature reference loss as a function of the number of switchable lines in each distribution network. The number of lines corresponds to the number of binary line-status variables in the MISOCP formulation. Unguided MISOCP denotes the standalone solver, while end-to-end runtime includes the LR-QAOA stage, pre- and post-processing, and the subsequent guided MISOCP optimization. Markers indicate the median across 20 solver seeds, with error bars spanning the interquartile range.
}
    \label{fig:runtime-comparison}
\end{figure}

Fig.~\ref{fig:417-gurobi} examines the 417-bus case in more detail by comparing the incumbent line-loss trajectories of the guided and unguided MISOCP solves. The guided solver begins with a substantially lower incumbent because the LR-QAOA sample distribution provided a high-quality feasible configuration before branch-and-bound begins in Gurobi. As a result, the median guided trajectory reaches the $1\%$ target much earlier than the unguided trajectory, which must first discover a comparable incumbent through the standalone MISOCP search. Over longer solve times, both methods continue improving toward the literature reference loss, but the guided run provides high-quality incumbents earlier in the optimization process. 

\begin{figure}[htbp]
    \centering
    \includegraphics[width=\linewidth]{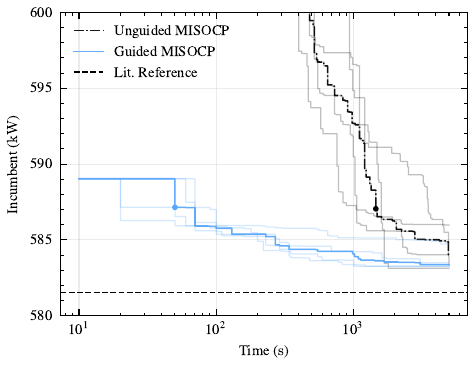}
    \caption{Incumbent line-loss trajectories for the guided and unguided MISOCP solvers on the 417-bus network. Faint lines show five individual runs for each solver. The solid blue and dash-dot black lines show the corresponding median trajectories for the guided and unguided solvers, respectively. Markers indicate the first median incumbent within 1\% of the literature reference loss. The dashed horizontal line denotes the literature reference loss of 581.57 $\mathrm{kW}$.
}
    \label{fig:417-gurobi}
\end{figure}

Taken together, these results show that the LR-QAOA-derived distribution provides useful information for the MISOCP solver. Across all six test systems, the guided runs reach high-quality incumbents earlier than the unguided runs, with the largest improvements occurring for the larger instances in which the standalone solve requires more time to identify comparable solutions. With the current quantum hardware, however, the end-to-end runtime remains dominated by the LR-QAOA sampling stage for several networks. The results therefore establish the effectiveness of QPU-derived solver guidance under present hardware constraints, while identifying quantum-execution overhead as the principal barrier to an end-to-end runtime advantage.

\section{Conclusion}\label{sec:conclusions}

In this work, we developed a two-stage quantum-classical framework for distribution network reconfiguration that combines iterative LR-QAOA sampling with a distribution-guided MISOCP solver. A central contribution is the proposed cycle-edge encoding, which restricts the encoded search space to spanning-tree configurations and therefore satisfies radiality and connectivity by construction while allowing the encoded subspace to be adapted according to a qubit budget. Iteratively reconstructing the encoded subspace around updated reference configurations enables different regions of the topologically feasible space to be explored without increasing the number of qubits.

Across six distribution networks, iterative LR-QAOA sampling identifies high-quality configurations under both ideal simulation and trapped-ion quantum hardware execution. For the largest tested system, the 417-bus network, the ideal simulation identifies a configuration within $1\%$ of the literature reference loss, while the QPU execution reaches within approximately $1.4\%$. Using the hardware-derived distributions to guide the MISOCP solver reduced median solver time to reach an incumbent within $1\%$ of the literature reference loss across all six systems, with the largest reduction occurring for the larger instances. Although quantum-stage overhead prevents an end-to-end runtime advantage on the tested systems, the relative runtime gap narrows for the larger instances, motivating evaluation on larger and more computationally challenging networks.  

Future work should examine larger networks and qubit budgets to determine whether increased subspace coverage improves sample quality and end-to-end performance. Larger encoded subspaces could expose configurations unavailable to the 29-qubit circuits used here, potentially reducing the number of iterations required to identify low-loss solutions and, for some instances, the need for the subsequent MISOCP solve to reach the target. Testing this behavior on more difficult instances, where the unguided solver requires substantially longer to reach high-quality solutions, would help establish whether the additional hardware quantum resources translate into greater overall benefit. 

The subspace-construction procedure could also be improved by evaluating multiple alternative edges per block or by allowing feasible MISOCP incumbents to define reference configurations for subsequent LR-QAOA iterations, creating a more tightly coupled quantum-classical search. More broadly, understanding which properties of LR-QAOA distributions provide the most useful variable hints could lead to improved mechanisms for translating quantum samples into classical solver guidance and help identify regimes in which such guidance provides the greatest benefit for large-scale DNR.

\appendices 
\crefalias{section}{appendix}
\crefalias{subsection}{appendix}

\section{Proofs of the Cycle-Edge Encoding Properties}\label{app:cycle-edge} 

\subsection{Spanning-Tree-Subspace Reachability}\label{app:spanning-tree-reachability}
\begin{lemma}
\label{lem:tree-contained}
Let $\mathcal{C}(T)$ denote the set of fundamental cycles associated with a spanning tree $T$. Since every fundamental-cycle basis is a simple-cycle basis, it is an admissible cycle basis under the proposed encoding. Let $\mathcal{W}(\mathcal{C}'(T))$ denote the encoded subspace associated with
a valid partition $\mathcal{C}'(T)$. There exists a valid partition $\mathcal{C}'(T)$ such that $T\in\mathcal{W}(\mathcal{C}'(T))$.
\end{lemma}

\begin{proof}
Let $O(T)=E\setminus T=\{o_1,o_2,\ldots,o_k\}$ denote the non-tree edges of $T$. Each $o_i$ belongs exclusively to its
corresponding fundamental cycle $C_i$. Therefore, under any valid partition
of the cycle-edge sets,
$
    o_i\in C_i',\text{ for } i=1,\ldots,k.
$
Selecting $o_i$ from each partition block opens exactly the non-tree edges
of $T$. Hence, the resulting configuration is $T$, and therefore $T\in\mathcal{W}(\mathcal{C}'(T))$.
\end{proof}

\par\noindent 
\textbf{Theorem \ref{thm:reachability}} (restated).
\textit{Let $\mathcal{T}(G)=\{T_1,T_2,\ldots,T_t\}$ denote the set of all spanning trees of a connected
graph $G$. There exists a finite collection of valid partitions
\begin{align}
    &\mathcal{P} = \{\mathcal{C}'_1,\mathcal{C}'_2,\ldots,\mathcal{C}'_m\}\\
    \text{s.t.}\; &
    \mathcal{W}(\mathcal{P}) := \bigcup_{i=1}^{m} \mathcal{W}(\mathcal{C}'_i) = \mathcal{T}(G).
\end{align}}

\begin{proof}[Proof of Theorem \ref{thm:reachability}]
By Lemma \ref{lem:tree-contained}, for every $T_i\in\mathcal{T}(G)$, there exists a valid partition $\mathcal{C}'(T_i)$ such that
\[
T_i\in\mathcal{W}(\mathcal{C}'(T_i)).
\]
Therefore, letting
\[ 
\mathcal{P} = \{\mathcal{C}'(T):T\in\mathcal{T}(G)\}, 
\]
it follows that
\begin{equation}
    \mathcal{W}(\mathcal{P}) = \bigcup_{\mathcal{C}'\in\mathcal{P}} \mathcal{W}(\mathcal{C}') = \mathcal{T}(G).
\end{equation}
\end{proof}

\section{Trapped-ion quantum hardware}\label{app:quantum_hardware} 

The quantum experiments in this work were performed on IonQ's Forte Enterprise quantum processing unit (QPU)~\cite{Chen2024-ik}. The system uses a register of 36 trapped $^{171}\mathrm{Yb}^{+}$ ions, with qubits encoded in hyperfine levels of the electronic ground state. The ions are generated through laser ablation and photoionization and confined in a surface-electrode linear Paul trap within an integrated vacuum system.

Qubit operations are driven by pulsed $\SI{355}{nm}$ laser fields that induce two-photon Raman transitions, providing arbitrary single-qubit rotations and native two-qubit $R_{ZZ}$ entangling gates. Individual-ion addressing is achieved using acousto-optic deflectors (AODs), which independently steer the control beams and reduce addressing and alignment errors across the ion chain~\cite{Kim2008-uj,Pogorelov2021-pu}. Automated calibration routines further support stable device operation throughout the experiments. Typical gate durations are approximately $\SI{130}{\micro s}$ for single-qubit operations and $\SI{950}{\micro s}$ for two-qubit entangling gates.

\section*{Acknowledgments}
\noindent This paper was prepared by EPB Quantum using Federal funds under award GONANB24D219-1 from the National Institute of Standards and Technology (NIST), U.S. Department of Commerce. The statements, findings, conclusions, and recommendations are those of the authors. This material is based upon work supported by the U.S. Department
of Energy, Office of Science, Office of Advanced Scientific Computing
Research under Award Number 89243024SSC000129 and under field work
proposal ERKJ445. N.A. and C.K. thank Noah Crum, Tanner Rase, Pete Pritchard, and Mason Blanchard of EPB for their insightful discussions that informed the theoretical framework, and David Nordy of EPB for guidance in power-flow analysis.  P.C.L. and S.R. thank Daniel Claudino and Teja Kuruganti for assistance in early development of the problem formulation.

\bibliographystyle{IEEEtran}
\bibliography{ref}
\end{document}